%% file: allertonfinal.tex
\documentclass[10pt, conference]{ieeeconf}      % Use this line for a4
\IEEEoverridecommandlockouts                              % This command is only
\usepackage{color}
\usepackage{psfrag}
\usepackage{graphics} % for pdf, bitmapped graphics files
\usepackage{epsfig} % for postscript graphics files
\usepackage{amsfonts,amsmath,amssymb} 
\usepackage{kbordermatrix}
\usepackage{cite}
\usepackage[mathcal]{euscript}
\newcommand{\bHt}{\bH_\mrm{t}}
\newcommand{\bHe}{\bH_\mrm{e}}
\newcommand{\bHr}{\bH_\mrm{r}}
\newcommand{\Nt}{n_\mrm{t}}
\newcommand{\Ne}{n_\mrm{e}}
\newcommand{\Nr}{n_\mrm{r}}
\newcommand{\bOmeg}{{\bar {\boldsymbol{\Omega}}}}

\title{\LARGE \bf The MIMOME Channel}
\author{Ashish Khisti and Gregory Wornell % <-this % stops a space
\thanks{This work was supported in part by NSF under Grant
No.~CCF-0515109.}  \thanks{The authors are with the Dept. EECS, MIT,
  Cambridge, MA, 
02139. Email:\{khisti,gww\}@mit.edu} }

\input{gww_defs}

\input{gww_chars}

\input{gww_rvs}

\begin{document}

\maketitle

\pagestyle{empty}

%%%%%%%%%%%%%%%%%%%%%%%%%%%%%%%%%%%%%%%%%%%%%%%%%%%%%%%%%%%%%%%%%%%%%%%%%%%%%%%%
\begin{abstract}
The MIMOME channel is a Gaussian wiretap channel in which the sender,
receiver, and eavesdropper all have multiple antennas.  We
characterize the secrecy capacity as the saddle-value of a minimax
problem.  Among other implications, our result establishes that a
Gaussian distribution maximizes the secrecy capacity characterization
of Csisz{\'a}r and K{\"o}rner when applied to the MIMOME channel. We
also determine a necessary and sufficient condition for the secrecy
capacity to be zero. Large antenna array analysis of this condition
reveals several useful insights into the conditions under which secure
communication is possible.
\end{abstract}

\section{Introduction}
Multiple antennas are a valuable resource in wireless communications. Recently there has been a significant activity in exploring both the theoretical and practical aspects of wireless systems with multiple antennas. In this work we explore the role of multiple antennas for physical layer security, which is an emerging  
area of interest. 

The wiretap channel\cite{wyner:75Wiretap} is an information theoretic model for physical layer security. The setup has three terminals --- one sender, one receiver and one eavesdropper. The goal is to exploit the structure of the underlying broadcast channel to transmit a message reliably to the intended receiver, while leaking asymptotically no information to the eavesdropper. A single letter characterization of the secrecy capacity, when the underlying  channel is a discrete memoryless broadcast channel, has been obtained by Csisz{\'a}r and K{\"o}rner~\cite{csiszarKorner:78}. An explicit solution for the scalar Gaussian case is obtained in~\cite{leung-Yan-CheongHellman:99}. 

In this paper we consider the case where all the three terminals have
multiple antennas and naturally refer to it as multiple input,
multiple output, multiple eavesdropper (MIMOME) channel. In this setup
we assume that the channel matrices are fixed and known to all the
three terminals. While the assumption that the eavesdropper's channel
is known to both the sender and the receiver is obviously a strong
assumption, we remark in advance that our solution provides ultimate
limits on secure transmission with multiple antennas and could be a
starting point for other formulations where the eavesdropper's channel
may not be known to the sender and the receiver.

The main result of this paper is a characterization of the secrecy capacity of the MIMOME channel as the saddle value of a minimax problem. Our approach does not rely on the Csisz{\'a}r and K{\"o}rner capacity expression, but instead is based on the technique used in characterizing the sum rate of the MIMO broadcast channel (see, e.g.,~\cite{yuwei:06} and its references). We first develop a minimax expression that upper bounds  the secrecy capacity and subsequently establish the tightness of this bound for the MIMOME channel.

The case where the channel matrices of intended receiver and
eavesdropper are square and diagonal follows from the results
in~\cite{liangPoor07,liYatesTrappe:06a,khistiTchamWornell:07,gopalaLaiElGamal:06Secrecy}
that consider secure transmission over fading channels. The difficulty
of optimizing the Csisz{\'a}r and K{\"o}rner expression for the
general case has been reported
in~\cite{NegiGoel05,LiTrappeYates07,shaifeeUlukus:07} and achievable
rates have been investigated. The approach used in the present paper
has been used in our earlier
work~\cite{khistiWornellEldar:07,khistiWornell:07} to establish the
secrecy capacity for two special cases: the case when the intended
receiver has a single antenna (MISOME case) and the MIMOME secrecy
capacity in the high SNR regime. This upper bounding approach was
independently conceived by Ulukus et. al.~\cite{UlukusCorresp} and
further applied to the 2x2x1 case~\cite{shafieeLiuUlukus:07}.
Finally, a related approach for the MIMOME channel, is developed
independently in~\cite{FrederiqueHassibi:07}. Also it is interesting
to note that this upper bounding approach has been empirically
observed to be tight for the problem of broadcasting two private
messages to two receivers when each receiver has a single
antenna~\cite{liuPoor:07}.  For this setup a single letter
characterization is not known for the discrete memoryless
case~\cite{LiuMaricSpasojevicYates:07,cai:06}

\section{Channel Model}
We denote the number of antennas at the sender, the receiver and the eavesdropper by $\Nt$, $\Nr$ and $\Ne$ respectively. \begin{equation}
\begin{aligned}
\rvby_\mrm{r}(t) &= \bHr\rvbx(t) + \rvbz_\mrm{r}(t)\\
\rvby_\mrm{e}(t) &= \bHe\rvbx(t) + \rvbz_\mrm{e}(t),
\end{aligned}
\end{equation}
where $\bHr \in \compls^{\Nr\times \Nt}$ and $\bHe \in \compls^{\Ne\times \Nt}$ are channel matrices associated with the receiver and the eavesdropper. The channel matrices are fixed for the entire transmission period and known to all the three terminals. The additive noise $\rvbzr(t)$ and $\rvbze(t)$  are circularly-symmetric and complex-valued Gaussian random variables. The input satisfies a power constraint $E\left[\frac{1}{n}\sum_{t=1}^n ||\rvbx(t)||^2\right]\le P.$

A rate $R$ is achievable if there exists a sequence of length $n$
codes, such that the error probability at the intended receiver and
$\frac{1}{n}I(\rvw;\rvby_\mrm{e}^n)$ both approach zero as
$n\rightarrow \infty$.  The secrecy capacity is the supremum of all
achievable rates.

\section{MIMOME Secrecy Capacity}

Our main result is the following characterization of the secrecy
capacity of the MIMOME wiretap channel.
\begin{thm}
The secrecy capacity of the MIMOME wiretap channel is
\begin{equation}
C = \min_{\bK_\bPh \in \cK_\bPh}\max_{\bKP\in \cK_\mrm{P}}R_+(\bKP,\bK_\bPh),
\label{eq:capacity}
\end{equation}
where $R_+(\bKP,\bK_\bPh)= I(\rvbx;\rvby_\mrm{r} \mid \rvby_\mrm{e})$ with $\rvbx \sim \CN(\mathbf{0},\bKP)$ and 
\begin{equation}
\begin{aligned}
\cK_\mrm{P} \defeq \left\{\bKP \Biggm| \bKP \succeq {\bf 0},\quad
\tr(\bKP) \le P \right\},
\end{aligned}\label{eq:KP-def}
\end{equation} and where $[\rvbzr^\dagger, \rvbze^\dagger]^\dagger \sim \CN(\mathbf{0},\bK_\bPh)$, with \begin{equation} 
\begin {aligned}
\cK_\bPh &\defeq \left\{\bK_\bPh \Biggm| \bK_\bPh =
\begin{bmatrix} \bI_{n_\mrm{r}} & \bPh \\ \bPh^\dagger & \bI_{n_\mrm{e}} \end{bmatrix},\quad \bK_\bPh \succeq {\bf 0}
\right\} \\
&= \left\{\bK_\bPh \Biggm| \bK_\bPh =
\begin{bmatrix} \bI_{n_\mrm{r}} & \bPh \\ \bPh^\dagger & \bI_{n_\mrm{e}} \end{bmatrix},\quad
\sigma_\mrm{max}(\bPh)\le1\right\}. 
\end {aligned}
\label{eq:Kph-def}
\end{equation}
Furthermore,\footnote{In the remainder of this paper, $\bI$  denotes
  an identity matrix and $\mathbf{0}$  denotes the matrix with all zeros. The
  dimensions of these matrices  will be suppressed and will be clear
  from the context. Also we use the superscript $^\dagger$ to denote
  the  hermitian conjugate of a matrix.   } the minimax problem
in~\eqref{eq:capacity} has a saddle point solution $(\obKP,\obKPh)$
and the secrecy capacity can also be expressed as, 
\begin{equation}
C = R_+(\obKP,\obKPh) = \log\frac{\det(\bI + \bHr\obKP\bHr^\dagger)}{\det(\bI + \bHe\obKP\bHe^\dagger)}.\label{eq:alternateCap}
\end{equation}
\label{thm1}
\end{thm}

\subsection{Connection with Csisz{\'a}r and K{\"o}rner Capacity}
A characterization of the secrecy capacity for the non-degraded discrete memoryless broadcast channel $p_{\rvy_\mrm{r},\rvy_\mrm{e}|\rvx}$ is provided by Csisz{\'a}r and K{\"o}rner~\cite{csiszarKorner:78},
\begin{align} 
C = \max_{p_{\rvu},p_{\rvx|\rvu}}
I(\rvu;\rvy_\mrm{r})-I(\rvu;\rvy_\mrm{e}),
\label{eq:CK}
\end{align}
where $\rvu$ is an auxiliary random variable (over a certain alphabet with bounded cardinality)
that satisfies $\rvu \rightarrow \rvx \rightarrow
(\rvy_\mrm{r},\rvy_\mrm{e})$.   As remarked
in~\cite{csiszarKorner:78}, the secrecy capacity \eqref{eq:CK} can be
extended in principle to incorporate continuous-valued inputs.
However, directly identifying the optimal $\rvu$ for the MIMOME case
is not straightforward.   

Theorem~\ref{thm1} indirectly establishes an optimal choice of $\rvu$ in~\eqref{eq:CK}.  Suppose that $(\obKP,\obKPh)$ is a saddle point solution to the minimax problem in~\eqref{eq:capacity}. From~\eqref{eq:alternateCap} we have
\begin{equation}
R_+(\obKP,\obKPh)=R_-(\obKP),\label{eq:SaddlePointProperty}
\end{equation}
where $$R_-(\obKP)\defeq \log\frac{\det(\bI +\bHr\obKP\bHr^\dagger)}{\det(\bI
  + \bHe\obKP\bHe^\dagger)}$$
is the achievable rate  obtained by
evaluating \eqref{eq:CK} for $\rvu=\rvx \sim \CN(\mathbf{0},\obKP)$. This
choice of $p_\rvu, p_{\rvbx|\rvu}$ thus
maximizes~\eqref{eq:CK}. Furthermore note that 
\begin{equation}
\label{eq:achievRate}
\begin{aligned}
\obKP &\in  \argmax_{\bKP \in \cK_\mrm{P}}\log\frac{\det(\bI + \bHr \bKP \bHr^\dagger)}{\det(\bI + \bHe \bKP \bHe^\dagger)}
\end{aligned}
\end{equation}
 where the set $\cK_\mrm{P}$ is defined in~\eqref{eq:KP-def}. Unlike
 the minimax problem~\eqref{eq:capacity} the maximization
 problem~\eqref{eq:achievRate} is not a convex optimization problem
 since the objective function is not a concave function of
 $\bKP$. Even if one verifies that $\obKP$ satisfies the optimality
 conditions associated with~\eqref{eq:achievRate}, this will only
 establish that $\obKP$ is a locally optimal solution. The capacity
 expression ~\eqref{eq:capacity} provides a convex reformulation
 of~\eqref{eq:achievRate} and establishes that $\obKP$ is a globally
 optimal solution in~\eqref{eq:achievRate}.\footnote{The ``high SNR"
 case of this problem i.e., $\max_{\bK \in \cK_\mrm{\infty}}\log
 \frac{\det(\bHr\bK \bHr^\dagger)}{\det(\bHe\bK\bHe^\dagger)}$ is known as the
 multiple-discriminant-function in multivariate statistics and is
 well-studied; see, e.g.,~\cite{Wilks:62}.}

\subsection{Structure of the optimal solution}
As we establish in Section~\ref{subsec:SaddleValue}, if
$(\obKP,\obKPh)$ is a saddle point solution to the minimax problem, if
$\bS$ is any matrix that has a  full column rank matrix and satisfies
$\obKP=\bS\bS^\dagger$ and if $\oPh$ is the cross-covariance matrix
between the noise random variables (c.f.~\eqref{eq:Kph-def}), then 
\begin{equation}
\bHe\bS = \oPh^\dagger \bHr \bS.\label{eq:degradedness}
\end{equation}
The condition in~\eqref{eq:degradedness} admits an intuitive interpretation. From~\eqref{eq:Kph-def} $\oPh$ is a contraction matrix i.e., all its singular values are less than or equal to unity. The column space of  $\bS$ is the subspace in which the sender transmits information. So~\eqref{eq:degradedness} states that no information is transmitted along any direction where the eavesdropper observes a stronger signal than the intended receiver.  The effective channel of the eavesdropper, $\bHe\bS$, is a degraded version of the effective channel of the intended receiver, $\bHr\bS$.

\section{Proof of Theorem~\ref{thm1}}
Our proof involves two main parts. First we show that the right hand side in~\eqref{eq:capacity} is an upper bound on the secrecy capacity. Then we examine the optimality conditions associated with the saddle point solution to establish~\eqref{eq:SaddlePointProperty}, which completes the proof since 
$$C \le R_+(\obKP, \obKPh) = R_-(\obKP)\le C.$$

That the right hand side in~\eqref{eq:capacity} is an upper bound on the secrecy capacity has already been established:
\begin{lemma}[\cite{khistiWornell:07,khistiWornellEldar:07}] An upper bound on the secrecy capacity for the MIMOME channel is 
\begin{equation}
C \le \min_{\bK_\bPh \in \cK_\bPh}\max_{\bKP\in \cK_\mrm{P}}R_+(\bKP,\bK_\bPh),
\end{equation}
where $\cK_\mrm{P}$ and $\cK_\bPh$ are defined in~\eqref{eq:KP-def} and~\eqref{eq:Kph-def} respectively.
\end{lemma}

Hence it suffices to establish~\eqref{eq:SaddlePointProperty}, which we do in the remainder of this section. We divide the proof into several steps, which are outlined in Fig.~\ref{fig:proof}.
\begin{figure}
\small
\psfrag{Sdl}{Saddle Point: $(\obKP,\obKPh)$}
\psfrag{Kcond}{$\displaystyle \obKPh \in \argmin_{\!\cK_\bPh}\!R_+(\!\obKP,\!\bK_\bPh\!)$}
\psfrag{Qcond}{$\displaystyle\obKP \in \argmax_{ \cK_P} R_+(\bKP,\obKPh)$}
\psfrag{Qcond2}{$\displaystyle\obKP \in \argmax_{ \cK_P} h(\rvbyr-\obTh\rvbye)$}
\psfrag{FinalCondn}{$\oPh^\dagger \bH_\mrm{r} \bS =  \bH_\mrm{e} \bS \Rightarrow R_+(\obKP,\obKPh)=R_-(\obKP)$}
\includegraphics[scale=0.34]{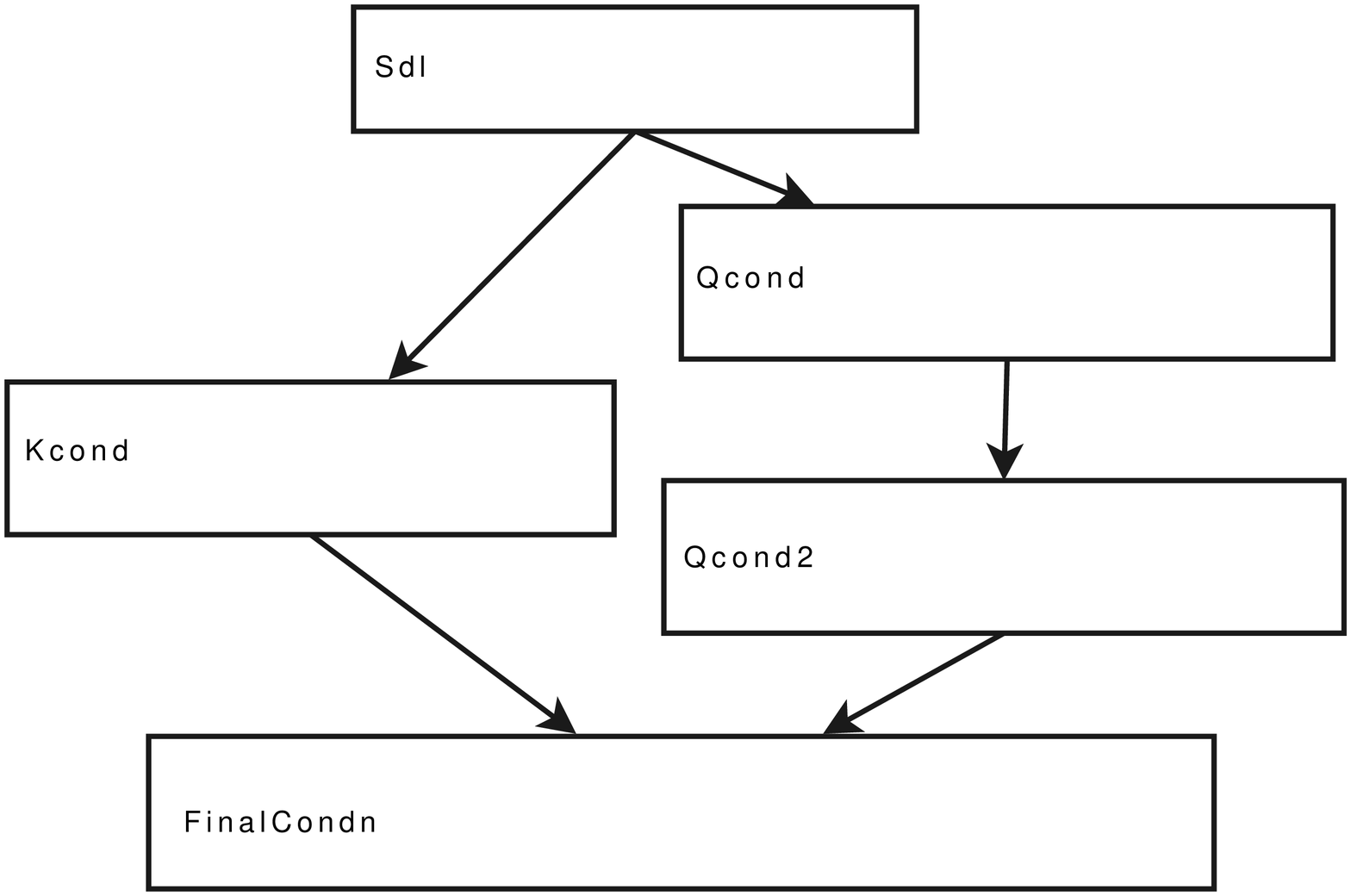}
\caption{Key steps in the Proof of Theorem 1. In section~\ref{subsec:saddlePoint} we establish that the minimax problem has a saddle point $(\obKP,\obKPh)$ . Section~\ref{subsec:optNoise} obtains a  condition satisfied by $(\obKP,\obKPh)$ via  the KKT conditions associated with the noise covariance, while section~\ref{subsec:optCov} obtains another condition that $(\obKP,\obKPh)$ satisfy, by first showing that $\obKP$ is also an optimal covariance of another MIMO channel. Combining these two conditions we show in Section~\ref{subsec:SaddleValue}  that the upper and lower bounds coincide.   }
\label{fig:proof}
\end{figure}

\subsection{Existence of the Saddle Point}
\label{subsec:saddlePoint}
Our first step is to show that for the minimax problem in~\eqref{eq:capacity}, a saddle point solution exists, i.e., there exists a point $(\obKP,\obKPh)$ with $\obKP \in \cK_\mrm{P}$  and $\obKPh \in \cK_\bPh$, such that for any $\bKP\in \cK_\mrm{P}$ and $\bK_\bPh\in \cK_\bPh$, we have that
\begin{equation}R_+(\bKP,\obKPh) \le R_+(\obKP,\obKPh)\le R_+(\obKP,\bK_\bPh).\label{eq:RKbKphSaddle}\end{equation}

Towards this end, we show the following convexity properties of the objective function.
% whose proof will be provided in the full version.
\begin{claim}
\label{claim:concavity}
For any fixed \mbox{$\bKP \in \cK_\mrm{P}$}, the function $R_+(\bKP,\bK_\bPh)$ is convex in $\bK_\bPh$. For any fixed $\bK_\bPh \in \cK_\bPh$, the function $R_+(\bKP,\bK_\bPh)$ is concave in $\bKP$. 
\end{claim}

\begin{proof}
Recall  that  $R_+(\bKP,\bK_\bPh)=I(\rvbx;\rvbyr,\rvbye)-I(\rvbx;\rvbye)$, with  $\rvbx \sim \CN(0,\bKP)$ and $[\rvbzr^\dagger~~\rvbze^\dagger]^\dagger \sim \CN(0, \bK_\bPh)$. For the convexity in $\bK_\bPh$, note that, $I(\rvbx;\rvbye)$ does not depend on $\bK_\bPh$, and $I(\rvbx;\rvbyr,\rvbye)$ is  known (see e.g.,~\cite{diggaviCover01}) to be convex in $\bK_\bPh$. For the concavity in $\bKP$, note that when $\bK_\bPh \succ \mathbf{0}$, 
we can express 

\vspace{-1em}

\begin{equation}R_+(\bKP,\bK_\bPh) = \log\det {{\bLa(\bKP)}} -
  \log\det\bK_\bPh, \label{eq:RSchurDef}\end{equation} 
where

\vspace{-1em}

\begin{multline}
\bLa(\bKP) \defeq \bI + \bHr\bKP\bHr^\dagger  -\\
(\bPh + \bHr \bKP\bHe^\dagger) (\bI +\bHe\bKP\bHe^\dagger)^{-1}(\bPh^\dagger + \bHe\bKP\bHr^\dagger) \label{eq:lamDef}
\end{multline} 
is the Schur compliment of the matrix\begin{equation}\begin{bmatrix}\bI + \bHr\bKP\bHr^\dagger & \bPh + \bHr \bKP\bHe^\dagger \\ \bPh^\dagger + \bHe\bKP\bHr^\dagger & \bI +\bHe\bKP\bHe^\dagger\end{bmatrix}.\label{eq:fullMatrix}\end{equation}
Since the Schur complement is jointly concave in the constituent
matrices \cite[page 21, Corollary 1.5.3]{bhatia:PSD},
which in turn are linear in $\bKP$, it follows that $\bLa(\bKP)$  is
concave in $\bKP$ and hence from the composition theorem we have that
$R_+(\bKP,\bK_\bPh)$ is concave\footnote{The concavitiy result can also be established via~\cite[pg. 506, Theorem 16.9.1]{coverThomas}, by observing that $\bI + \bHe\bKP\bHe^\dagger$ is a minor of the matrix in~\eqref{eq:fullMatrix}.  } in $\bKP$.  The case when $\bK_\bPh$ is singular, can be handled via the singular value decomposition of $\bPh$, and will be treated in the full paper.
\end{proof}

\begin{comment}
\begin{proof}
The convexity in $\bK_\bPh$ for a fixed $\bKP$ follows from the fact that $$R_+(\bKP,\bK_\bPh) =I(\rvbx;\rvbyr,\rvbye)-I(\rvbx;\rvbye),$$
and only the term $I(\rvbx;\rvbyr,\rvbye)$ depends on $\bK_\bPh$,
which in turn is known to be convex in $\bK_\bPh$ (see,
e.g.,~\cite{diggaviCover01}). For the concavity in $\bKP$, note that
we can express 
\begin{equation}R_+(\bKP,\bK_\bPh) = \log\det {{\bLa(\bKP)}} -
  \log\det\bK_\bPh,\end{equation} 
where
\begin{multline}
\bLa(\bKP) \defeq \bI + \bHr\bKP\bHr^\dagger  -\\
(\bPh + \bHr \bKP\bHe^\dagger) (\bI +\bHe\bKP\bHe^\dagger)^{-1}(\bPh^\dagger + \bHe\bKP\bHr^\dagger) \label{eq:lamDef}
\end{multline} 
is the Schur compliment of the matrix
$$\begin{bmatrix}\bI + \bHr\bKP\bHr^\dagger & \bPh + \bHr \bKP\bHe^\dagger \\ \bPh^\dagger + \bHe\bKP\bHr^\dagger & \bI +\bHe\bKP\bHe^\dagger\end{bmatrix}.$$
Since the Schur complement is jointly concave in the constituent
matrices \cite[page 21, Corollary 1.5.3]{bhatia:PSD},
which in turn are linear in $\bKP$, it follows that $\bLa(\bKP)$  is
concave in $\bKP$ and hence from the composition theorem we have that
$R_+(\bKP)$ is concave in $\bKP$. 
\end{proof}

\end{comment}
Notice that both the domain sets $\cK_\mrm{P}$ and $\cK_\bPh$
are convex and compact, hence the existence of a saddle point solution
$(\obKP,\obKPh)$ is established via the minimax theorem \cite{bertsekasNedicOzdaglar}.  
%We now show that any saddle point solution $(\obKP,\obKPh)$ satisfies  certain  %conditions which in turn imply~\eqref{eq:SaddlePointProperty}.

In the sequel, we define $\oPh$ via
\begin{equation} \obKPh =
\begin{bmatrix}\bI_{n_\mrm{r}} & \oPh\\\oPh^\dagger & \bI_{n_\mrm{e}}\end{bmatrix}.
\label{eq:oPh-def}
\end{equation}

\subsection{Least favorable noise condition}
\label{subsec:optNoise}

From~\eqref{eq:RKbKphSaddle}, we have that
\begin{equation}
\obKPh \in \arg\min_{\bK_\bPh \in \cK_\bPh}R_+(\obKP,\bK_\bPh).\label{eq:noiseMin}
\end{equation}
The optimality conditions associated with~\eqref{eq:noiseMin} yield the following.
\begin{lemma}
Suppose that $(\obKP,\obKPh)$ is a saddle point solution to the
minimax problem in~\eqref{eq:capacity}. Then
\begin{equation}
(\bHr -\obTh\bHe)\obKP(\oPh^\dagger\bHr
    -\bHe)^\dagger = \mathbf{0}.
\label{eq:NoiseCondn}
\end{equation}
where $\oPh$ is as defined via \eqref{eq:oPh-def} and 
\begin{equation} 
\obTh   = (\bHr\obKP\bHe^\dagger + \oPh)(\bI + \bHe\obKP\bHe^\dagger)^{-1}.
\label{eq:obTh-def}
\end{equation}
\label{lem:NoiseCondn} 
\end{lemma}

We will see subsequently, that~\eqref{eq:NoiseCondn} has a useful structure, which can be combined with the optimality condition associated with $\obKP$. The proof  is most direct when the noise covariance $\obKPh$ at the saddle point is non-singular. Hence we will establish~\eqref{eq:NoiseCondn} in this special case first and then consider the case  when $\obKPh$ is singular.

\subsubsection{$\obKPh$ is non-singular.}
The Lagrangian associated with the minimization~\eqref{eq:noiseMin} is \begin{equation}\cL_\bPh(\bK_\bPh, \bUp) = R_+(\obKP,\bK_\bPh) + \tr(\bUp \bK_\bPh), \end{equation}
where the dual variable
\begin{equation}\label{eq:UpsDefn}\bUp=
{\kbordermatrix{& n_\mrm{r} & n_\mrm{e} \cr
                        n_\mrm{r} & \bUp_1 & \mathbf{0} \\
n_\mrm{e}  & \mathbf{0} & \bUp_2}}
\end{equation}
is a block diagonal matrix corresponding to the constraint that the
noise covariance $\bK_\bPh$ must have identity matrices on its
diagonal. The associated Karush-Kuhn-Tucker (KKT) conditions
yield
%\begin{subequations}
\begin{equation*} 
\nabla_{\bK_\bPh}R_+(\obKP,\bK_\bPh)\bigr|_{\obKPh} + \bUp =
  \mathbf{0},
\end{equation*}
where
\begin{align} 
\nabla_{\bK_\bPh} &R_+(\obKP,\bK_\bPh)\bigr|_{\obKPh} \\
&= \nabla_{\bK_\bPh}\left[\log\det(\bK_\bPh +\bHt\obKP\bHt^\dagger)\!-\!\log\det(\bK_\bPh)\right]\biggr|_{\obKPh}\nonumber\\
&=(\obKPh + \bHt\obKP\bHt^\dagger)^{-1} - \obKPh^{-1},
\label{eq:noiseKKTCondn}
\end{align}
with the
convenient notation
\begin{equation} 
\bHt = \begin{bmatrix} \bHr \\ \bHe\end{bmatrix},
\label{eq:bHt-def}
\end{equation}
which in turn implies that
\begin{equation} 
\bHt\obKP\bHt^\dagger = \obKPh\bUp(\obKPh + \bHt\obKP\bHt^\dagger).
\label{eq:noiseKKTCondnb}\end{equation}
%\end{subequations}
The relation in~\eqref{eq:NoiseCondn} follows from~\eqref{eq:noiseKKTCondnb} through a straightforward computation that exploits the block diagonal structure of $\bUp$, which we provide in Appendix~\ref{app:NoiseRel}.

\subsubsection{$\obKPh$ is singular} When the noise covariance $\obKPh$ is singular, as we now show, ~\eqref{eq:noiseKKTCondnb} still holds. Note that this will complete the proof, since the steps in Appendix~\ref{app:NoiseRel} that simplify~\eqref{eq:noiseKKTCondnb} do not require that $\obKPh$ be non-singular. 

In the singular case we define another optimization problem whose optimality conditions yield~\eqref{eq:noiseKKTCondnb}. An analogous approach has been taken earlier by Yu~\cite{yuwei:06} for dealing with singular noise for the MIMO broadcast channel.
 
Suppose that \begin{equation}
\obKPh= \bW \bOmeg\bW^\dagger,
\label{eq:Wdecomp}
\end{equation}
where $\bW$ is a matrix with orthogonal columns, i.e., $\bW^\dagger\bW = \bI$ and $\bOmeg$ is a non-singular matrix. We first note that it must also be the case that
\begin{equation}
\label{eq:HGWrel}
\bHt = \bW \bG,
\end{equation}
i.e., the column space of $\bHt$ is a subspace of the column space of $\bW$. If this were not the case, by receiving a signal in the null space of $\bW$, one can obtain arbitrarily high rate, i.e., 
\begin{equation}
\max_{\bKP\in\cK_\mrm{P}}R_+(\bKP,\obKPh) = \infty,
\end{equation}
which contradicts that $\obKPh$ is a saddle point solution.

Now observe that $\bOmeg$ in~\eqref{eq:Wdecomp} is a solution to the following minimization problem, 
\begin{equation}
\begin{aligned}
&\min_{\bOm \in \cOm }R_\Omega(\bOm ),\\
&\quad R_\Omega(\bOm ) = \log\frac{\det(\bG\obKP\bG^\dagger +\bOm) }{\det(\bOm )},\\
&\quad \cOm = \left\{\bOm \Biggm| \bW \bOm \bW^\dagger = \begin{bmatrix}\bI_{n_r} & \bPh \\ \bPh^\dagger & \bI_{n_e} \end{bmatrix} \succeq 0 \right\}.
\end{aligned}\label{eq:noiseSingularMinProblem}
\end{equation}
Indeed $\bOmeg$ is a feasible point
for~\eqref{eq:noiseSingularMinProblem}. Also with
$\rvbz_\Omega\sim\CN(\mathbf{0},\Omega)$, one can show that  
\begin{align}R_\Omega(\bOm)=
 R_+(\obKP,\bW\Omega\bW^\dagger) + \log\det(\bI + \bHe\obKP\bHe^\dagger)\label{eq:R+RomegRel},\end{align}
from which the optimality of $\bOmeg$ readily follows. 
\begin{comment}
\begin{align}R_\Omega(\bOm)&= I(\rvbx; \bG\rvbx + \rvbz_\Omega)\nonumber\\
&=I(\rvbx; \bW\bG\rvbx + \bW\rvbz_\Omega)\label{eq:Wmult}\\
&=I(\rvbx; \bHt\rvbx + \rvbz) \label{eq:omNoiseProp}\\
&=I(\rvbx; \rvbyr,\rvbye)\nonumber\\
&= R_+(\obKP,\bW\Omega\bW^\dagger) + \log\det(\bI + \bHe\obKP\bHe^\dagger)\label{eq:R+RomegRel}\end{align}
where~\eqref{eq:Wmult} follows from the fact that $\bW$ has a full column rank, so right multiplication by $\bW$ preserves the mutual information and~\eqref{eq:omNoiseProp} follows from~\eqref{eq:HGWrel} and the fact that $\rvbz_\Omega \in \cOm$. If there exists any $\Wh$ such that $R_\Omega(\Wh) <R_\Omega(\bOmeg)$ then from~\eqref{eq:R+RomegRel} we have that $$R_+(\obKP,\obKPh)> R_+(\obKP,\bW\Wh\bW^\dagger)$$ which (c.f.~\eqref{eq:RKbKphSaddle}) contradicts that $\obKPh$ is a saddle point solution.
\end{comment}
The optimality conditions associated with the minimization problem~\eqref{eq:noiseSingularMinProblem} give
\begin{equation}
\begin{aligned}
\bOmeg^{-1}-(\bG\obKP\bG^\dagger + \bOmeg)^{-1}=\bW^\dagger\bUp\bW,\\
\Rightarrow \bG\obKP\bG^\dagger = \bOmeg\bW^\dagger \bUp \bW(\bOmeg + \bG\obKP\bG^\dagger)
\end{aligned}\label{eq:OmegKKT}\end{equation}
where $\bUp$ has the block diagonal form in~\eqref{eq:UpsDefn}. Multiplying the left and right and side of~\eqref{eq:OmegKKT} with $\bW$ and $\bW^\dagger$ respectively and using ~\eqref{eq:Wdecomp} and~\eqref{eq:HGWrel} we have that
\begin{equation}
\bHt\obKP\bHt^\dagger = \obKPh\bUp (\obKPh + \bHt\obKP\bHt^\dagger),
\end{equation}
which coincides with ~\eqref{eq:noiseKKTCondnb}.

\subsection{Optimal Input Covariance Property}
\label{subsec:optCov}

Given that $(\obKP,\obKPh)$ is a saddle point solution in~\eqref{eq:capacity} we have from~\eqref{eq:RKbKphSaddle} that
\begin{equation}
\obKP \in \argmax_{\bK_\mrm{P} \in \cK_\mrm{P}}R_+(\bK_\mrm{P},\obKPh).
\label{eq:KpOpt}
\end{equation}
We show that~\eqref{eq:KpOpt} in turn implies the following property.
\begin{lemma}
\label{lem:OptCovCondn}
Suppose that $\obKP = \bS\bS^\dagger$, where
$\bS$ has a full column rank.  Then
provided $(\bHr-\obTh\bHe) \neq \mathbf{0}$, the
  matrix\begin{equation}\label{eq:Mdef}\bM
  =(\bHr-\obTh\bHe)\bS\end{equation} has a full column rank,
where $\oPh$ and $\obTh$ are defined via \eqref{eq:oPh-def} and
\eqref{eq:obTh-def}, respectively.
\end{lemma}

The rest of this subsection is devoted to the proof of Lemma~\ref{lem:OptCovCondn}, and accordingly we assume that the saddle point solution $(\obKP,\obKPh)$ satisfies  $(\bHr-\obTh\bHe) \neq \mathbf{0}$.  As with Lemma~\ref{lem:NoiseCondn}, the proof is most direct when $\obKPh$ is non-singular. Hence we will treat this case first and consider the case when $\obKPh$ is singular subsequently.

\subsubsection{$\obKPh$ is non-singular}
In this case, we can write the optimality condition~\eqref{eq:KpOpt} as
\begin{align}
\obKP &\in \argmax_{\bK_\mrm{P} \in \cK_\mrm{P}}R_+(\bK_\mrm{P},\obKPh)\nonumber\\
&=\argmax_{\bK_\mrm{P} \in \cK_\mrm{P}}h(\rvbyr\mid \rvbye)\nonumber\\
&=\argmax_{\bK_\mrm{P} \in \cK_\mrm{P}}h\left(\rvbyr - \Theta(\bKP) \rvbye\right),\label{eq:mmseCondn1}
\end{align}
where $\Theta(\bKP)= (\bHr\bKP\bHe^\dagger + \oPh)(\bHe\bKP\bHe^\dagger+\bI)^{-1}$ is the linear minimum mean squared estimation coefficient of $\rvbyr$ given $\rvbye$. Instead of directly working with the optimality conditions associated with~\eqref{eq:mmseCondn1} we reformulate the problem as below.
\begin{claim}
\label{claim:HSaddlePoint}
Suppose that $\obKPh\succ {\bf 0}$ and define 
\begin{equation}
\begin{aligned}
&\cH(\bK_\mrm{P}) \defeq h(\rvbyr-\obTh\rvbye)
=\log\det(\bGa(\bKP)),
\end{aligned}\label{eq:Hdefn}\end{equation}
where
\begin{multline}
\bGa(\bK_\mrm{P}) \defeq \bI + \obTh\obTh^\dagger -\obTh\oPh^\dagger-\oPh\obTh^\dagger +\\ (\bHr-\obTh\bHe)\bKP(\bHr-\obTh\bHe)^\dagger.
\label{eq:GammaDef}
\end{multline}
Then,  %$\obKP$ is a solution to the following problem,
\begin{equation}
\obKP \in \argmax_{\bKP \in \cK_\mrm{P}}\cH(\bKP).
\label{eq:secondOpt}
\end{equation}
\end{claim}

\begin{remark}
The objective function in~\eqref{eq:secondOpt} is  similar to the one in~\eqref{eq:mmseCondn1}, but with $\obTh$ fixed, i.e., the variables  $\obTh$ and $\bKP$ are decoupled in~\eqref{eq:secondOpt}. This key step enables us to work with the simpler objective function in~\eqref{eq:secondOpt} and complete the proof. 
\end{remark}

\begin{proof}
To establish~\eqref{eq:secondOpt} note that since $\cH(\cdot)$ is a concave function in $\bK_\mrm{P} \in \cK_\mrm{P}$
and differentiable over  $\cK_\mrm{P}$, the optimality conditions associated with the Lagrangian
\begin{equation}
\cL_\Theta(\bK_\mrm{P},\lambda,\bPs) = \cH(\bK_\mrm{P})+ \tr(\bPs\bK_\mrm{P})-\lambda(\tr(\bK_\mrm{P})-P),
\end{equation}
are both necessary and sufficient. Thus $\bK_\mrm{P}$ is an optimal solution to~\eqref{eq:secondOpt}
if and only if there exists a $\lambda \ge 0$ and $\bPs \succeq 0$ such that\begin{equation}
\begin{aligned}
&(\bHr-\obTh\bHe)^\dagger[\bGa(\bK_\mrm{P})]^{-1}(\bHr-\obTh\bHe) + \bPs = \lambda\bI,\\
&\tr(\bPs \bK_\mrm{P})=0, \quad \lambda(\tr(\bK_\mrm{P})-P)=0,\label{eq:HmaxKKT}
\end{aligned}
\end{equation}
where $\bGa(\cdot)$ is defined in~\eqref{eq:GammaDef}.
These parameters  for $\obKP$ are obtained from the optimality conditions associated with~\eqref{eq:KpOpt}.

Since $R_+(\bK_\mrm{P},\obKPh)$ is differentiable at each
$\bK_{P} \in \cK_\mrm{P}$ whenever $\obKPh \succ {\bf 0}$, $\obKP$ satisfies the associated KKT conditions --- there exists a $\lambda_0 \ge 0$ and $\bPs_0 \succeq 0$ such that
\begin{equation}\begin{aligned}
\nabla_{\bK_\mrm{P}}R(\bK_\mrm{P},\obKPh)\Biggm|_{\obKP} +\bPs_0 =\lambda_0\bI\\
\lambda_0(\tr(\obKP)-P)=0,\quad \tr(\bPs_0 \obKP)=0.
\end{aligned}\label{eq:OptKpKKT}\end{equation}
We show in Appendix~\ref{app:KKT} that
\begin{equation}
\nabla_{\bK_\mrm{P}}R(\bK_\mrm{P},\obKPh)\Bigr|_{\obKP} \!=\!
(\bHr-\obTh\bHe)^\dagger[\bLa(\obKP)]^{-1}(\bHr-\obTh\bHe),
\end{equation} 
%\begin{multline}\text{where }\quad \!\bLa(\bKP)\defeq \bI +
%\bHr \bKP\bHr^\dagger - \\(\bHr\bKP\bHe^\dagger+\oPh)(\bHe\bKP\bHe^\dagger %+\bI)^{-1}(\bHr\bKP\bHe^\dagger+\oPh)^\dagger,\label{eq:LambdaDef}\end{multline}
%and  
where 
%\begin{multline}
%\bLa(\bKP) \defeq \bI + \bHr\bKP\bHr^\dagger  -\\
%(\bPh + \bHr \bKP\bHe^\dagger) (\bI +\bHe\bKP\bHe^\dagger)^{-1}(\bPh^\dagger + %\bHe\bKP\bHr^\dagger) \label{eq:lamDef}
%\end{multline} 
$\bLa(\cdot)$,  defined in~\eqref{eq:lamDef}, satisfies
satisfies\footnote{To verify this relation, note that $\bGa(\bKP)$ is the variance of $\rvbyr-\obTh \rvbye$. When  $\bKP=\obKP$, note that $\obTh\rvbye$ is the MMSE estimate of  $\rvbyr$ given $\rvbye$   and $\bGa(\bKP)$ is the associated MMSE estimation error.} $\bLa(\obKP)=\bGa(\obKP)$. Hence the first condition
in~\eqref{eq:OptKpKKT} reduces to 
\begin{equation}
(\bHr-\obTh\bHe)^\dagger[\bGa(\obKP)]^{-1}(\bHr-\obTh\bHe) + \bPs_0 = \lambda_0\bI.
\label{eq:OptKpKKTa}
\end{equation}
Comparing~\eqref{eq:OptKpKKT} and~\eqref{eq:OptKpKKTa} with~\eqref{eq:HmaxKKT}, we note that
$(\obKP,\lambda_0,\bPs_0)$ satisfy the conditions in~\eqref{eq:HmaxKKT}, thus establishing~\eqref{eq:secondOpt}.
\end{proof}

\begin{claim}
Suppose that $\obKPh \succ {\bf 0}$ and $\hbKP$ be any optimal solution to
\begin{equation}
\hbKP \in \argmax_{\cK_\mrm{P}}\cH(\bKP).
\label{eq:KPMax}\end{equation} Suppose that $\bSP$ is a matrix with a
full column rank such that
\begin{equation}\hbKP =\bSP\bSP^\dagger\label{eq:FullRankSP}\end{equation}  then 
$(\bHr-\obTh\bHe)\bSP$ has a full column rank. \label{claim:FullRankCase}
\end{claim}
%\begin{remark}
Note that the claim in Lemma~\ref{lem:OptCovCondn} follows from Claim~\ref{claim:HSaddlePoint} and Claim~\ref{claim:FullRankCase}. It remains to prove Claim~\ref{claim:FullRankCase}.
%\end{remark}

\begin{proof}
The proof is  based on the so called water-filling
principle~\cite{coverThomas}. From~\eqref{eq:KPMax}, we have
\begin{align}
&\hbKP=\notag\\
&\argmax_{\bKP\in \cK_\mrm{P}}\log\det(\bI\! +
  \!\bJ^{-\frac{1}{2}}(\bHr\!-\!\obTh\bHe)\bKP(\bHr\!-\!\obTh\bHe)^\dagger\bJ^{-\frac{1}{2}}),\label{eq:KpMax2} 
\end{align}
where $\bJ \defeq \bI +
\obTh\obTh^\dagger-\obTh\oPh^\dagger-\oPh\obTh^\dagger \succ {\bf 0}$, 
i.e., $\hbKP$ is an optimal input covariance for a MIMO channel with white noise and matrix
$\bH_\mrm{eff}\defeq \bJ^{-\frac{1}{2}}(\bHr-\obTh\bHe)$. We can now consider the usual water-filling properties associated with $\hbKP$ to establish that $(\bHr-\obTh\bHe)\bSP$ has a full column rank. 

Let $\mrm{rank}(\bH_\mrm{eff})=\nu$ and let us denote the non-zero singular values (in non-increasing order) by $\sigma_1,\sigma_2,\ldots,\sigma_\nu$. Let $\bSi_0 = \diag(\sigma_1,\ldots,\sigma_\nu)$, and
\begin{equation}
\bSi = {\kbordermatrix{& \nu & n_\mrm{t}-\nu \cr
                        \nu & \bSi_0 & \mathbf{0} \\
n_\mrm{r}-\nu  & \mathbf{0} & \mathbf{0}}},
\end{equation}
be such that
\begin{equation}
\bH_\mrm{eff} = \bA \bSi \bB^\dagger = \bA_1 \bSi_0 \bB_1^\dagger, 
\end{equation}
is the singular value decomposition of $\bH_\mrm{eff}$ where $\bA$ and $\bB$ are unitary matrices in $\mathbb{C}^{n_\mrm{r}\times n_\mrm{r}}$ and $\mathbb{C}^{n_\mrm{t}\times n_\mrm{t}}$ and
\begin{equation}
\bA = \kbordermatrix{& \nu & n_\mrm{r}-\nu \cr & \bA_1 & \bA_2},\quad
\bB=\kbordermatrix{& \nu & n_\mrm{t}-\nu \cr & \bB_1 & \bB_2}.
\end{equation}
From~\eqref{eq:KpMax2} we have that
\begin{align}
\hbKP &\in \argmax_{\cK_P}\log\det(\bI + \bH_\mrm{eff}\bKP\bH_\mrm{eff}^\dagger)\nonumber\\
&= \argmax_{\cK_P}\log\det(\bI + \bA\bSi\bB^\dagger \bKP\bB \bSi^\dagger \bA^\dagger)\nonumber\\
&=\argmax_{\cK_P}\log\det(\bI + \bB^\dagger\bKP\bB \bSi^\dagger\bSi)\label{eq:KpBMax}
\end{align}
Since $\bB$ is unitary, we have that $\bB^\dagger\bKP\bB  \in \cK_P$ and hence it follows from~\eqref{eq:KpBMax} that
\begin{align}
\bF \defeq \bB^\dagger\hbKP\bB \in \argmax_{\cK_P}\log\det(\bI + \bKP\bSi^\dagger\bSi).\label{eq:FMatDefn}
\end{align}
We now show that any such $\bF$ is diagonal and $\bF_{ii}=0$ for
$i>\nu$. From the Hadamard inequality~\cite[Section
  16.8]{coverThomas}, we have that 
\begin{equation} 
\log\det(\bI + \bF\bSi^\dagger\bSi)\!\le\!\sum_{i=1}^{n_\mrm{t}}\log(1 +
\bF_{ii}\sigma_{i}^2)\!=\!\sum_{i=1}^{\nu}\log(1 +
\bF_{ii}\sigma_{i}^2),
\label{eq:Hadamard}
\end{equation}
with equality if and only if  the matrix $\bF\bSi^\dagger\bSi$ is a diagonal matrix. We now show that any optimal $\bF$ in~\eqref{eq:FMatDefn} has the form
\begin{equation}
\bF = \kbordermatrix{ & \nu & n_\mrm{t} - \nu \\ \nu & \bF_0 & \mathbf{0} \\ n_\mrm{t}-\nu & \mathbf{0} & \mathbf{0}}\label{eq:Fstruct}
\end{equation}
where $\bF_0$ is a diagonal matrix. Clearly any optimal $\bF$ attains the upper bound in~\eqref{eq:Hadamard}, hence it follows that (1) $\sum_{i=1}^\nu \bF_{ii} = P$, and $\bF_{ii} = 0$ for $i> \nu$ and (2) $\bF \bSi^\dagger\bSi $ is a diagonal matrix. The first condition, together with the fact that $\bF\succeq \mathbf{0}$ imples that the lower diagonal matrix in \eqref{eq:Fstruct} is zero, while the second condition implies that the off-diagonal matrices in \eqref{eq:Fstruct} are zero and that $\bF_0$ is  diagonal.

\begin{comment}
where the last
equality follows from the fact that $\mrm{rank}(\bH_\mrm{eff})=\nu$,
so that $\sigma_{\nu+1}=\ldots=\sigma_{n_\mrm{t}}=0$.  Note this
implies that $\bF_{ii}=0$ for  $i>\nu$, since otherwise, one can
strictly increase the value of the objective function with another
$\bFt$ where $\bFt_{ii}=0$, $\bFt_{11}=\bF_{11}+\bF_{ii}$, and
$\bFt_{jj}=\bF_{jj}$ otherwise.  For the first inequality to be an
equality, we must also have that $\bF\bSi^\dagger\bSi$ 
be diagonal.  Using $\bF \succeq \mathbf{0}$ and $\bF_{ii}=0$ for $i>\nu$,  $\bF$ must have the form
\begin{equation}
\bF = \kbordermatrix{ & \nu & n_\mrm{t} - \nu \\ \nu & \bF_0 & \mathbf{0} \\ n_\mrm{t}-\nu & \mathbf{0} & \mathbf{0}}
\end{equation}
where $\bF_0$ is a diagonal matrix.  
\end{comment}
From~\eqref{eq:FMatDefn}, we have that
\begin{equation}
\hbKP =\bB\bF\bB^\dagger = \bB_1\bF_0\bB_1^\dagger\label{eq:hbKPDecomp}
\end{equation}
and hence for any $\bSP$ that has a full column rank and satisfies~\eqref{eq:FullRankSP}, we have
\begin{align*}
\mrm{col}(\bSP) &\subseteq \mrm{col}(\bB_1)
 =\mrm{Null}^\perp(\bH_\mrm{eff})=\mrm{Null}^\perp(\bHr-\obTh\bHe),
\end{align*}
which implies that $(\bHr-\obTh\bHe)\bSP$ has a full column rank. 
\end{proof}

\subsubsection{ $\obKPh$ is singular}
The case when $\obKPh$ is singular can be handled by considering an
appropriately reduced channel matrix. In this case $\oPh$ has $d\ge 1$
singular values equal to unity and hence we can express its SVD as
\begin{equation}
\oPh = \begin{bmatrix}\bU_1 & \bU_2\end{bmatrix}\begin{bmatrix}\bI & \mathbf{0}\\\mathbf{0} & \bDe\end{bmatrix}\begin{bmatrix}\bV_1^\dagger \\\bV_2^\dagger\end{bmatrix}
\label{eq:oPhSVD}
\end{equation}
where $\sigma_\mrm{max}(\bDe)< 1$.  

First we obtain some conditions that are satisfied when the saddle point noise covariance is singular.
\begin{claim}

Suppose that $(\obKP,\obKPh)$ is a saddle point solution to the minimax problem in~\eqref{eq:capacity} and the singular value decomposition of $\oPh$ is given as in
~\eqref{eq:oPhSVD}. Then we have that \begin{subequations}\begin{equation}
\bU_1^\dagger \rvbzr \stackrel{\mrm{a.s.}}{=}\bV_1^\dagger\rvbze\label{eq:noiseSing1}\end{equation}
\begin{equation}\bU_1^\dagger\bHr, = \bV_1^\dagger\bHe,\label{eq:noiseSing2}\end{equation}
\begin{equation}R_+(\bKP,\obKPh) = I(\rvbx;\bU_2^\dagger\rvbyr\mid \rvbye),\quad \forall~\bKP \in \cK_\mrm{P}.
\label{eq:noiseSing3}\end{equation} \end{subequations}\label{claim:SingularCondns}
\end{claim}

\begin{proof}
To establish~\eqref{eq:noiseSing1}, we simply note that
$$E[\bU_1^\dagger\rvbzr\rvbze^\dagger\bV_1]= \bU_1^\dagger\oPh\bV_1= \bI,$$
i.e., the Gaussian random variables $\bU_1^\dagger\rvbzr$ and $\bV_1^\dagger\rvbze$ are perfectly correlated. Next note that
\begin{align*}
R_+(\bKP,\obKPh) &= I(\rvbx;\rvbyr|\rvbye)\\
&=I(\rvbx; \bU_1^\dagger\rvbyr, \bU_2^\dagger\rvbyr|\rvbye)\\
&=I(\rvbx; \bU_2^\dagger\rvbyr, \bU_1^\dagger\rvbyr-\bV_1^\dagger\rvbye|\rvbye)\\
&=I(\rvbx; \bU_2^\dagger\rvbyr, \bU_1^\dagger\bHr\rvbx-\bV_1^\dagger\bHe\rvbx|\rvbye).
\end{align*}Since $\obKPh$ is a saddle point solution, we must have $\max_{\bKP}R_+(\bKP,\obKPh) < \infty$ and hence $\bU_1^\dagger\bHr=\bV_1^\dagger\bHe$,
and  $R_+(\obKP,\obKPh)=I(\rvbx;\bU_2^\dagger\rvbyr\mid\rvbye)$, establishing~\eqref{eq:noiseSing2} and~\eqref{eq:noiseSing3}. 
\end{proof}

Thus with ${\bHh_\mrm{r}}=\bU_2^\dagger\bHr$, and ${\hrvbzr}= \bU_2^\dagger\rvbz_\mrm{r}$ and \begin{equation}{\hrvbyr} = \bU_2^\dagger\rvbyr= {\bHh_\mrm{r}}\rvbx + {\hrvbzr},\label{eq:hatyreq}\end{equation} we have from \eqref{eq:noiseSing3}, that
\begin{align}
\obKP &\in  \argmax_{ \cK_\mrm{P}}I(\rvbx;{\hrvbyr}\mid\rvbye).\label{eq:SingularOpt1}
\end{align}
Since $\bPhh = E[{\hrvbzr}\rvbze^\dagger] \prec \bI$, it follows from~\eqref{eq:SingularOpt1} and Claim~\ref{claim:HSaddlePoint} that 
\begin{equation}
\obKP \in \argmax_{\cK_P}\cHh(\bKP)
\end{equation}
where
\begin{align*}
{\cHh}(\bKP) &= h({\hrvbyr}-\bThh\rvbye), \\
\bThh &=   \bU_2^\dagger(\bHr\obKP\bHe^\dagger + \oPh)(\bI +
  \bHe\obKP\bHe^\dagger)^{-1}.
\end{align*}
Along the lines of Claim~\ref{claim:FullRankCase} we then have that  $$({\bHh_\mrm{r}} - {\bThh} \bHe)\bS = \bU_2^\dagger(\bHr-\obTh\bHe)\bS$$ has a full column rank, which in turn establishes that $(\bHr-\obTh\bHe)\bS$ has a full column rank.

\subsection{Saddle Value}
\label{subsec:SaddleValue}
We use the results from Lemma~\ref{lem:NoiseCondn} and Lemma~\ref{lem:OptCovCondn} to establish~\eqref{eq:SaddlePointProperty}. To invoke Lemma~\ref{lem:OptCovCondn}, we will first assume that the saddle point solution $(\obKP,\obKPh)$ is such that $\bHr-\obTh\bHe\neq \mathbf{0}$ and treat the case  $\bHr-\obTh\bHe= \mathbf{0}$ subsequently. Note that from Lemma~\ref{lem:NoiseCondn} we have that\begin{equation}
(\bHr-\obTh\bHe)\bS \bS^\dagger (\oPh^\dagger\bHr-\bHe)^\dagger = \bf{0},\label{eq:noiseEqn}
\end{equation}
and since $\bM = (\bHr-\obTh\bHe)\bS$ has a full column rank,~\eqref{eq:noiseEqn} reduces to\begin{equation}\oPh^\dagger\bHr\bS=\bHe\bS.\label{eq:noiseEq1}\end{equation}

The difference between the upper and lower bounds is given by
\begin{align}
\bDe R &= R_+(\obKP,\obKPh)-R_-(\obKP)\nonumber\\
&= I(\rvbx;\rvbyr\mid\rvbye)-[I(\rvbx;\rvbyr)-I(\rvbx;\rvbye)]\nonumber\\
&= I(\rvbx;\rvbye\mid\rvbyr)\label{eq:diffR}.
\end{align}

If $\obKPh \succ {\bf 0}$, then $I(\rvbx;\rvbye\mid\rvbyr) = h(\rvbye\mid\rvbyr)-h(\rvbze\mid\rvbzr)
$ and\begin{align}
&h(\rvbye\mid\rvbyr)\nonumber\\
&= \log \det (\bI+\bHe\obKP\bHe^\dagger -  \nonumber\\
&\quad\quad(\bHe\obKP\bHr^\dagger + \oPh^\dagger)(\bHr\obKP\bHr^\dagger  + \bI)^{-1}(\bHr\obKP\bHe^\dagger+\oPh))\nonumber\\
&= \log \det (\bI+\bHe\obKP\bHe^\dagger- \oPh^\dagger(\bHr\obKP\bHr^\dagger + \bI)\oPh)\nonumber\\
&=\log\det(\bI-\oPh^\dagger\oPh)= h(\rvbze\mid\rvbzr)\label{eq:Leq},
\end{align}
where we have used the relation~\eqref{eq:noiseEq1}  in simplifying ~\eqref{eq:Leq}. This shows that the difference $\Delta R$ in
~\eqref{eq:diffR} is zero, thus establishing~\eqref{eq:SaddlePointProperty} whenever
$\obKPh$ is non-singular.

To establish the result when $\obKPh$ is singular, note that from \eqref{eq:noiseSing1} and~\eqref{eq:noiseSing2} in Claim~\ref{claim:SingularCondns},
\begin{align}
\Delta R &= I(\rvbx;\rvbye\mid\rvbyr),\nonumber\\
&= I(\rvbx;\bV_2^\dagger\rvbye\mid \rvbyr),\label{eq:diffRSingular}
\end{align} 
which is zero as shown below.
\begin{align}
&h(\bV_2^\dagger\rvbye\mid\rvbyr) \nonumber\\
&=\log\det(\bI + \bV_2^\dagger\bHe\obKP\bHe^\dagger\bV_2 -
  (\bV_2^\dagger \bHe\obKP\bHr^\dagger + \bDe^\dagger\bU_2^\dagger) \nonumber \\
& \qquad\qquad(\bI +\bHr\obKP\bHr^\dagger)^{-1}(\bHr\obKP\bHe^\dagger\bV_2 + \bU_2\bDe ))\label{eq:Leq1}\\
&=\log\det(\bI + \bDe^\dagger \bU_2^\dagger\bHr\obKP\bHr^\dagger
  \bU_2\bDe \nonumber\\ 
&\qquad\qquad-\bDe^\dagger\bU_2^\dagger(\bI + \bHr\obKP\bHr^\dagger)\bU_2\bDe)\nonumber\\
&=\log\det(\bI-\bDe^\dagger\bDe)\nonumber\\
&=h(\bV_2^\dagger\rvbze \mid \bU_2^\dagger
\rvbzr)=h(\bV_2^\dagger\rvbze\mid\rvbzr)\label{eq:Leq2},
\end{align}
where
we have used from~\eqref{eq:noiseEq1} that 
\begin{equation*}
\bV_2^\dagger\oPh^\dagger \bHr \bS = \bV_2^\dagger\bHe\bS \Rightarrow \bDe^\dagger\bU_2^\dagger\bHr \bS = \bV_2^\dagger\bHe\bS, %\label{eq:noiseEq2}
\end{equation*}in simplifying~\eqref{eq:Leq1} and the equality in~\eqref{eq:Leq2} follows from the fact that $\bU_1^\dagger \rvbzr$ is independent of $(\bU_2^\dagger\rvbzr,\bV_2^\dagger\rvbze)$. This establishes~\eqref{eq:SaddlePointProperty} when $\obKPh$ is singular.

It remains to consider the case when the saddle point solution $(\obKP,\obKPh)$ is such that\begin{equation}
\obTh\bHe = \bHr.\label{eq:zeroCondn}
\end{equation}
In this case, we show that the saddle value and hence the capacity is zero. From~\eqref{eq:obTh-def}, $\obTh = (\oPh + \bHr\obKP\bHe^\dagger)(\bI + \bHe\obKP\bHe^\dagger)^{-1}$, hence we have
\begin{equation}
\obTh + \obTh\bHe\obKP\bHe^\dagger = \oPh + \bHr\obKP\bHe^\dagger.
\label{eq:zeroCondn2}
\end{equation}
Substituting~\eqref{eq:zeroCondn} in~\eqref{eq:zeroCondn2}, we have that $\oPh=\obTh$, and using this relation it can be verified that $R_+(\obKP,\obKPh)=\mathbf{0}$.
This completes the proof of Theorem~\ref{thm1}.

\section{Zero-Capacity Condition and Scaling Laws}

The conditions on $\bHr$ and $\bHe$  for which the secrecy capacity is zero have a simple form.
\begin{lemma}
The secrecy capacity of the MIMOME channel is zero if and only if
\begin{equation}
\sigma_\mrm{max}(\bHr,\bHe) \defeq \sup_{\bv\in\mathbb{C}^{\Nt}}\frac{||\bHr\bv||}{||\bHe\bv||} \le 1.
\label{eq:largeEigValue}
\end{equation}
\label{lem:zeroCapCond}
\end{lemma}
We omit the proof of this condition due to space constraints. The quantity $\sigma_\mrm{max}(\bHr,\bHe)$ is the largest generalized singular value of the channel matrices \cite{golubVanLoan}.
Analysis of the zero-capacity condition in the limit of large number of antennas provides several useful insights we develop below.

For our analysis, we use the following convergence property of the largest generalized singular value for Gaussian matrices.

\begin{fact}[~\cite{silverStein85,BaiSilverstein:95}]
Suppose that $\bHr$ and $\bHe$ have i.i.d. $\CN(0,1)$ entries. Let
$n_\mrm{r},n_\mrm{e},n_\mrm{t}\rightarrow \infty$, while keeping
$n_\mrm{r}/n_\mrm{e}=\g$  and $n_\mrm{t}/n_\mrm{e}=\beta$ fixed. If
$\beta < 1$, then the largest generalized singular value of
$(\rvbH_\mrm{r},\rvbH_\mrm{e})$ converges almost surely to 
\begin{equation}\sigma_\mrm{max}(\rvbH_\mrm{r},\rvbH_\mrm{e})\stackrel{\mrm{a.s.}}{\rightarrow} {\g}\left[\frac{1 + \sqrt{1-(1-\beta)\left(1-\frac{\beta}{\g}\right)}}{1-\beta }\right]^2.\label{eq:f1eq}\end{equation}
\label{fact:LargestGSVConv}
\end{fact}

By combining Lemma~\ref{lem:zeroCapCond} and Fact~\ref{fact:LargestGSVConv}, one can deduce the following condition for the zero-capacity condition.

\begin{corol}
Suppose that $\rvbH_\mrm{r}$ and $\rvbH_\mrm{e}$ have i.i.d.~$CN(0,1)$
entries. Suppose that $n_\mrm{r},n_\mrm{e},n_\mrm{t} \rightarrow
\infty$, while keeping $n_\mrm{r}/n_\mrm{e}=\g$  and
$n_\mrm{t}/n_\mrm{e}=\beta$ fixed. The secrecy
capacity\footnote{We assume that the channels are sampled once, then
  stay fixed for the entire  period of transmission, and are revealed
  to all the terminals.} $C(\rvbH_\mrm{r},\rvbH_\mrm{e})$ converges
almost surely to zero if and only if $ 0 \le \beta \le 1/2$, $0 \le \g
\le 1$, and 
\begin{equation}
%\beta \le \frac{1}{2}\left(1 + \g - %2\sqrt{\g}\right).
\g \le (1-\sqrt{2\beta})^2.\label{eq:AsympZeroCondn}
\end{equation}\label{corol:AsympZeroCondn}
\end{corol}

\begin{figure*}
\begin{minipage}[b]{0.5\linewidth}
\includegraphics[scale=0.35]{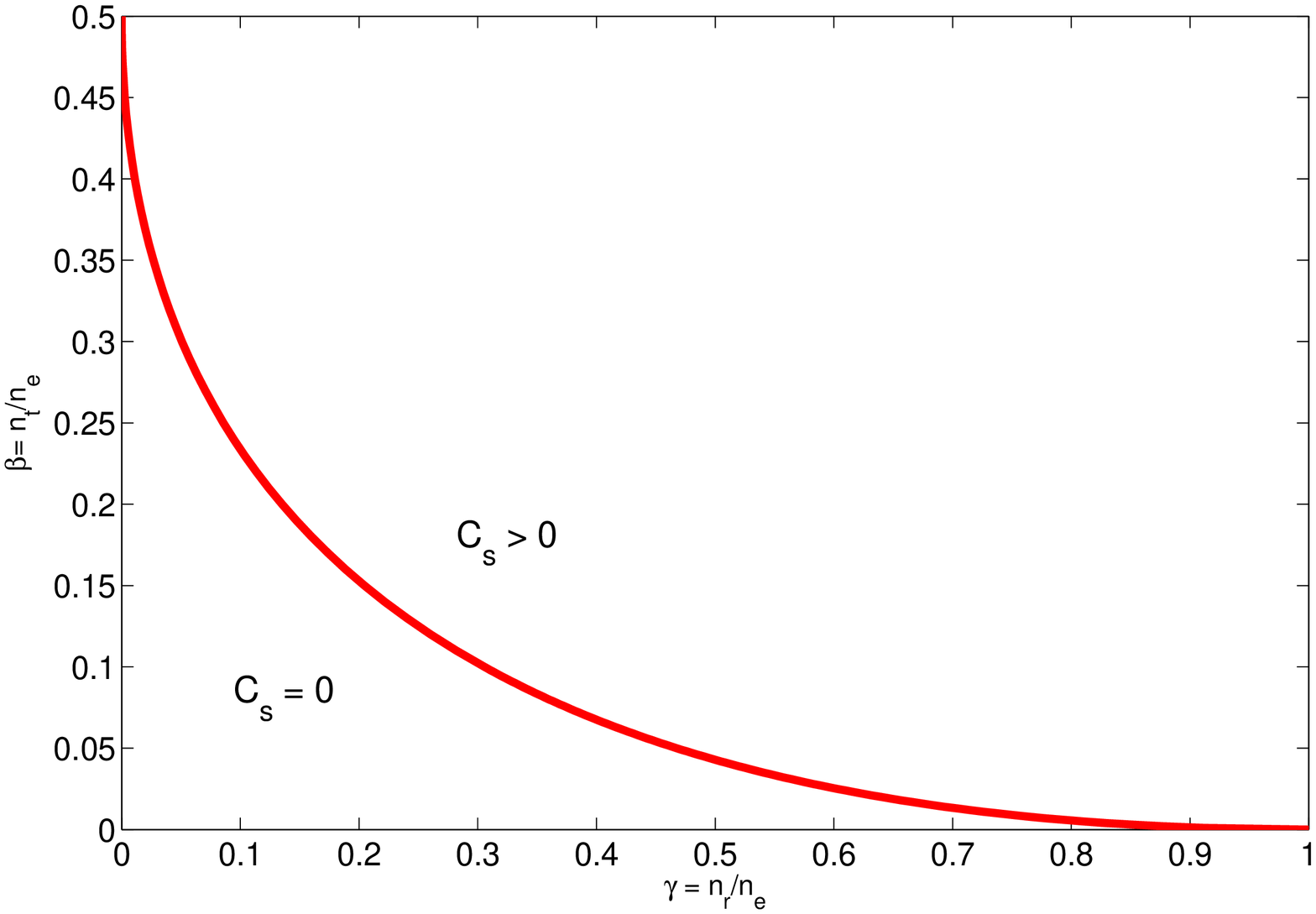}
\caption{Zero-capacity condition in the $(\beta,\gamma)$ plane.  The capacity is zero for any point below the curve, i.e., the eavesdropper has sufficiently many antennas to get non-vanishing fraction of the message, even when the sender and receiver fully exploit the knowledge of $\rvbH_\mrm{e}$. }
\label{fig:zeroPlot1}
\end{minipage}\textcolor{white}{|}
\begin{minipage}[b]{0.5\linewidth}
\includegraphics[scale=0.35]{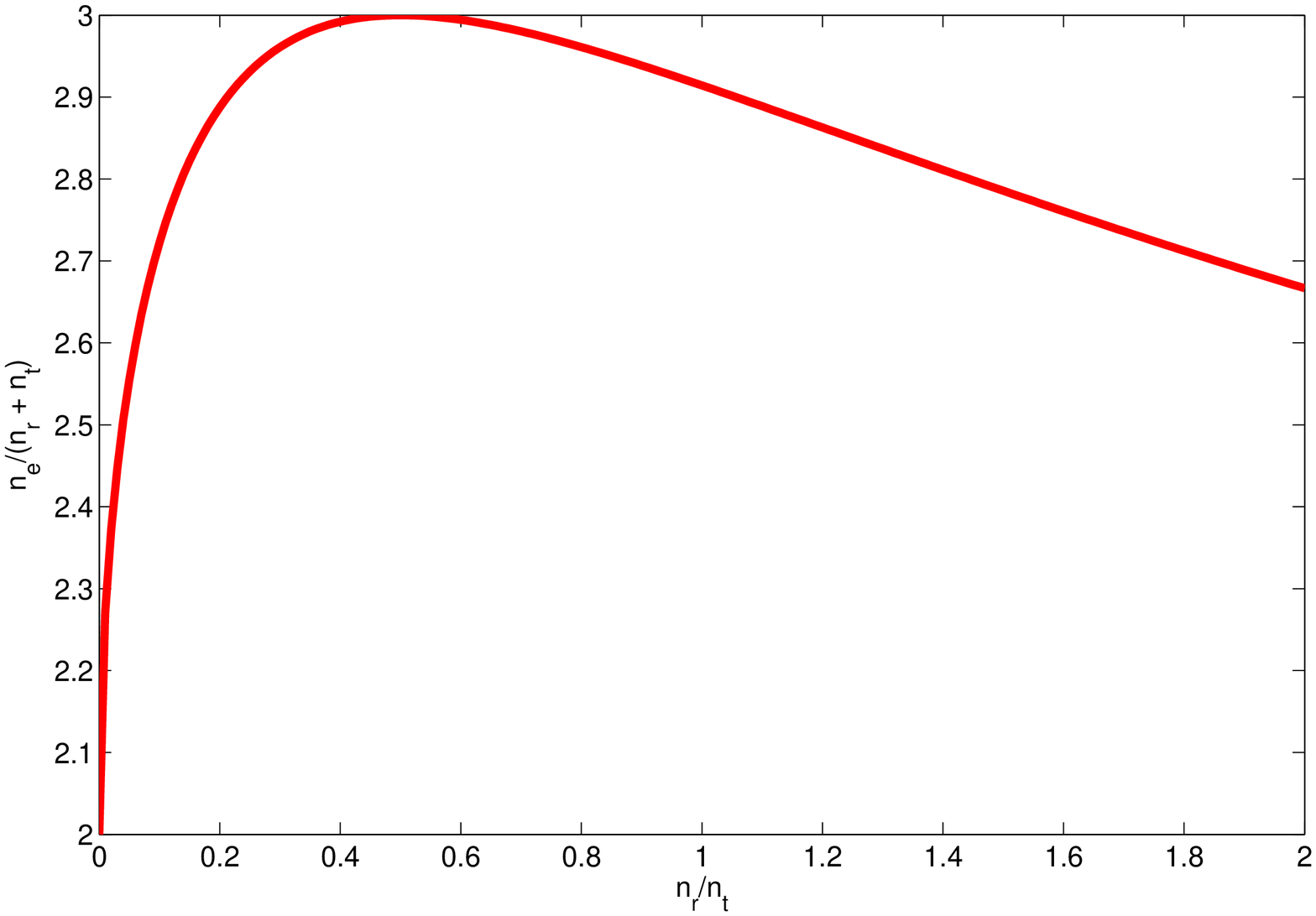}
\caption{The minimum number of eavesdropping antennas per sender plus receiver antenna for the secrecy capacity to be zero, plotted as a function of $n_\mrm{r}/n_\mrm{t}$.  }
\label{fig:zeroPlot2}
\end{minipage}
\end{figure*}

Figs.~\ref{fig:zeroPlot1} and~\ref{fig:zeroPlot2} provide further
insight into the asymptotic analysis for the capacity achieving
scheme. In Fig.~\ref{fig:zeroPlot1}, we show the values of
$(\beta,\gamma)$ where the secrecy rate is zero. If the eavesdropper
increases its antennas at a sufficiently high rate so that the point
$(\beta,\gamma)$ lies below the solid curve, then secrecy capacity is
zero.  The MISOME case corresponds to the vertical intercept of this
plot. The secrecy capacity is zero, if $\beta \le 1/2$, i.e., the
eavesdropper has at least twice the number of antennas as the sender.
The single transmit antenna (SIMOME) case corresponds to the
horizontal intercept. In this case the secrecy capacity is zero if $\g
\le 1$, i.e., the eavesdropper has more antennas than the receiver.

In Fig.~\ref{fig:zeroPlot2}, we consider the scenario where a total
of $T\gg 1$ antennas are divided between the sender and the
receiver. The horizontal axis plots the ratio $n_\mrm{r}/n_\mrm{t}$,
while the vertical axis plots the minimum number of antennas at the
eavesdropper (normalized by $T$) for the secrecy capacity to be
zero. We note that the optimal allocation of antennas, that maximizes
the number of eavesdropper antennas happens at $n_\mrm{r}/n_\mrm{t}
=1/2$. This can be explicitly obtained from the following minimization
\begin{equation}\begin{aligned} 
&\mrm{minimize}~~\bt + \g\\
&\mrm{ subject~to},~~\g \ge (1-\sqrt{2\bt})^2,~~\bt \ge 0,~~\g \ge 0.
\end{aligned}\label{eq:myOpt}\end{equation}

The optimal solution can be easily verified to be $(\beta^*,\g^*)= (2/9,1/9)$.
In this case, the eavesdropper needs $\approx 3T$  antennas for the secrecy capacity to be zero. We remark that the objective function in~\eqref{eq:myOpt} is not sensitive to
variations in the optimal solution. If fact even if we allocate equal 
number of antennas to the sender and the 
receiver, the eavesdropper needs $\frac{(3+2\sqrt{2})}{2}T \approx 2.9142\times T$ antennas for the secrecy capacity to be zero.

\begin{comment}
\section{Conclusion}
We establish the secrecy capacity of the MIMOME channel as a saddle point solution to a minimax problem. Our capacity result establishes that a Gaussian input
maximizes the secrecy capacity expression by Csisz{\'a}r and K{\"o}rner for the MIMOME channel. Our proof uses upper bounding ideas from the MIMO broadcast channel literature and the analysis of optimality conditions provides insight into the structure of the optimal solution.   The study of scaling laws provides a framework to obtain architectural insights into the role of multiple antennas for secure communications.

In terms of future work, it would be worthwhile to extend the insights obtained from the solution to the fixed channel matrices to cases when the eavesdropper's channel is not known. 
\end{comment}
\section*{Acknowledgement}
Ami Wiesel provided a numerical optimizer to evaluate the saddle point expression in Theorem~\ref{thm1}.

\begin{comment}
\subsection{Synthetic Noise based scheme}
Our capacity achieving scheme fully exploits the knowledge of $\rvbH_\mrm{e}$ at the sender and the receiver. In this section, we consider a suboptimal scheme that chooses the transmit directions only based on the knowledge of the intended receiver's channel. 

Our transmission scheme is as follows. We transmit the signal of interest along the strongest eigen-mode of the intended receiver's channel and transmits artificial noise in all orthogonal directions. Accordingly, if $\bv_\mrm{max}$ denotes the right singular vector corresponding to the largest singular value $\sigma_\mrm{max}$ of $\rvbH_\mrm{r}$, our choice of distribution in~\eqref{eq:CK} is
\begin{equation}p_\rvu=  \CN(0,\tilde{P}), \quad p_{\rvbx|u} =
  \CN\left(u\bv_\mrm{max},\tilde{P}(\bI -
  \bv_\mrm{max}\bv_\mrm{max}^\dagger)\right)\end{equation}
where
$\tilde{P}=P/n_\mrm{t}$ and the corresponding achievable rate is 
\begin{equation}
R_\mrm{MB}(P) = \log(1+\lambda_\mrm{max}\Pt)+ \log(\bv_\mrm{max}^\dagger(\bI + \Pt\rvbH_\mrm{e}^\dagger\rvbH_\mrm{e} )^{-1}\bv_\mrm{max}).
\end{equation}

$$f_2(\beta,\gamma)= \frac{\left(\sqrt{\g} + \sqrt{\beta}\right)^2}{1-\beta}$$

\end{comment}
\appendices

\section{Least Favorable Noise Property}
\label{app:NoiseRel}
Substituting for $\obKPh$ and $\bHt$ in~\eqref{eq:noiseKKTCondnb} and
carrying out the block matrix multiplication gives
\begin{equation}
\begin{aligned}
\bHr \obKP\bHr^\dagger &= \bUp_1(\bI + \bHr\obKP\bHr^\dagger) + \oPh\bUp_2(\oPh^\dagger + \bHe\obKP\bHr^\dagger)\\
\bHr\obKP\bHe^\dagger &= \bUp_1(\oPh + \bHr\obKP\bHe^\dagger)+\oPh\bUp_2(\bI + \bHe\obKP\bHe^\dagger)\\
\bHe\obKP\bHr^\dagger &= \oPh^\dagger\bUp_1(\bI + \bHr\obKP\bHr^\dagger)+ \bUp_2(\oPh^\dagger +\bHe\obKP\bHr^\dagger)\\
\bHe\obKP\bHe^\dagger &=\oPh^\dagger\bUp_1(\oPh + \bHr\obKP\bHe^\dagger) + \bUp_2(\bI + \bHe\obKP\bHe^\dagger).
\end{aligned}\label{eq:noiseKKTCondn2}\end{equation}

Eliminating $\bUp_1$ from the first and third equation above, we have
\begin{equation}
(\oPh^\dagger \bHr-\bHe)\obKP\bHr^\dagger = (\oPh^\dagger\oPh -\bI)\bUp_2(\oPh^\dagger +\bHe\obKP\bHr^\dagger).
\label{eq:ups21}
\end{equation}
Similarly eliminating $\bUp_1$ from the second and fourth equations in~\eqref{eq:noiseKKTCondn2} we have
\begin{equation}
(\oPh^\dagger\bHr-\bHe)\obKP\bHe^\dagger = (\oPh^\dagger\oPh-\bI)\bUp_2(\bI + \bHe\obKP\bHe^\dagger).
\label{eq:ups22}
\end{equation}
Finally, eliminating $\bUp_2$ from~\eqref{eq:ups21}
and~\eqref{eq:ups22} we obtain \eqref{eq:NoiseCondn}.

\section{KKT Condition}
\label{app:KKT}
\begin{comment}
Recall that the block matrix inversion lemma ~\cite{petersenPedersen}
has the following form 
\begin{align}
&\begin{bmatrix}
\bPs_{11} & \bPs_{12}\\
\bPs_{21} & \bPs_{22}
\end{bmatrix}^{-1} \notag\\
&\quad=\begin{bmatrix}
\bLa^{-1} & -\bLa^{-1}\bPs_{12}\bPs_{22}^{-1}\\
-\bPs_{22}^{-1}\bPs_{21}\bLa^{-1} & \bPs_{22}^{-1}+ \bPs_{22}^{-1}\bPs_{21}\bLa^{-1}\bPs_{12}\bPs_{22}^{-1}
\end{bmatrix},
\label{eq:inversion}
\end{align}
where $\bLa = \bPs_{11} -
\bPs_{12}\bPs_{22}^{-1}\bPs_{21}$
is the Schur complement of $\bPs$. 
\end{comment}
First note that,
\begin{equation}
\begin{aligned}
&\nabla_{\bK_\mrm{P}}R_+(\bK_\mrm{P},\obKPh) \\
&= \bHt^\dagger(\bHt\bKP\bHt^\dagger + \obKPh)^{-1}\bHt -
  \bHe^\dagger(\bI + \bHe\bKP\bHe^\dagger)^{-1}\bHe. \label{eq:nabla-exp}
\end{aligned}
\end{equation}

Substituting for $\bHt$ and $\obKPh$ from~\eqref{eq:bHt-def} and~\eqref{eq:oPh-def}, 
\begin{align*}
&(\obKPh + \bHt\obKP\bHt^\dagger)^{-1}\\
&=\begin{bmatrix}\bI + \bHr\obKP\bHr^\dagger & \oPh + \bHr \obKP \bHe^\dagger\\
\oPh^\dagger + \bHr\obKP\bHe^\dagger & \bI + \bHe\obKP\bHe^\dagger\end{bmatrix}^{-1}\\
&=\begin{bmatrix}\bLa^{-1} & -\bLa^{-1}\obTh\\
-\obTh^\dagger\bLa^{-1}  & (\bI\!+\!\bHe\obKP\bHe)^{-1}\!+\!\obTh^\dagger\bLa^{-1}\obTh\end{bmatrix}^{-1},
\end{align*}
where we have used the matrix inversion lemma (e.g.,~\cite{petersenPedersen}), and  $\bLa \defeq
\bLa(\obKP)$  is defined in~\eqref{eq:lamDef}, and  $\obTh$ is
as defined in~\eqref{eq:obTh-def}. Substituting into
\eqref{eq:nabla-exp} and simplifying gives
\begin{align*}
&\nabla_{\bK_\mrm{P}}R_+(\bK_\mrm{P},\obKPh)\biggm|_{\obKP}\\
=&\bHt^\dagger(\obKPh + \bHt\obKP\bHt^\dagger)^{-1}\bHt - \bHe^\dagger(\bI + \bHe\obKP\bHe^\dagger)^{-1}\bHe\\
&\qquad\qquad= (\bHr-\obTh\bHe)^{\dagger}[\bLa(\obKP)]^{-1}(\bHr-\obTh\bHe)
\end{align*}
as required. 

%\vspace{-2em}

\IEEEtriggeratref{13}
%\begin{small}
\bibliographystyle{IEEEtran}
\bibliography{sm}
%\end{small}

\end{document}

%% file: gww_defs.tex
\usepackage{verbatim} % get comment environment, + new verbatim
 \usepackage{amsxtra}  % get, eg, \accentedsymbol
\DeclareMathOperator*{\argmax}{arg\,max}
\DeclareMathOperator*{\argmin}{arg\,min}

\DeclareMathOperator{\diag}{diag}

\DeclareMathOperator{\tr}{tr}

\newcounter{actr}
{\begin{list}{(\alph{actr})}{\usecounter{actr}}}{\end{list}}

\newcounter{ictr}
{\begin{list}{(\roman{ictr})}{\usecounter{ictr}}}{\end{list}}

\newtheorem{remark}{Remark}
\newtheorem{thm}{Theorem}
\newtheorem{lemma}{Lemma}
\newtheorem{claim}{Claim}
\newtheorem{corol}{Corollary}

\newtheorem{fact}{Fact}
\newenvironment{new-proof}[1]
{{\em Proof  #1: }}%
{ \noindent\qed }
\newcommand{\qed}{\rule[0.1ex]{1.4ex}{1.6ex}}

\newcommand{\defeq}{\stackrel{\Delta}{=}}

\newcommand{\compls}{\mathbb{C}}

\newcommand{\mrm}{\mathrm}

%% file: gww_chars.tex
\newcommand{\bA}{{\mathbf{A}}}

\newcommand{\bB}{{\mathbf{B}}}

\newcommand{\bF}{{\mathbf{F}}}

\newcommand{\bFt}{{\tilde{\bF}}}

\newcommand{\bG}{{\mathbf{G}}}

\newcommand{\bH}{{\mathbf{H}}}
\newcommand{\bHh}{{\hat{\mathbf{H}}}}

\newcommand{\cH}{{\mathcal H}}
\newcommand{\cHh}{\hat{{\mathcal H}}}

\newcommand{\bI}{{\mathbf{I}}}

\newcommand{\bJ}{{\mathbf{J}}}

\newcommand{\bK}{{\mathbf{K}}}

\newcommand{\cK}{{\mathcal{K}}}
\newcommand{\bKP}{{\mathbf{K}}_\mathrm{P}}
\newcommand{\hbKP}{\hat{{\mathbf{K}}}_\mathrm{P}}
\newcommand{\obKP}{{\bar{\mathbf{K}}}_\mathrm{P}}
\newcommand{\obKPh}{{\bar{\mathbf{K}}}_{\Phi}}
\newcommand{\oPh}{{\bar{\boldsymbol{\Phi}}}}

\newcommand{\cL}{{\mathcal{L}}}

\newcommand{\bM}{{\mathbf{M}}}

\newcommand{\CN}{{\mathcal{CN}}}
\newcommand{\Pt}{{\tilde{P}}}

\newcommand{\bS}{{\mathbf{S}}}

\newcommand{\bSP}{{\mathbf{S}}_\mrm{P}}

\newcommand{\bU}{{\mathbf{U}}}

\newcommand{\bv}{{\mathbf{v}}}
\newcommand{\bV}{{\mathbf{V}}}

\newcommand{\bW}{{\mathbf{W}}}

\newcommand{\bt}{\beta}

\newcommand{\g}{\gamma}

\newcommand{\bGa}{{\boldsymbol{\Gamma}}}

\newcommand{\De}{\Delta}

\newcommand{\bDe}{{\boldsymbol{\De}}}

\newcommand{\bPh}{{\boldsymbol{\Phi}}}
\newcommand{\bPhh}{{\boldsymbol{\hat{\Phi}}}}

\newcommand{\bPs}{{\boldsymbol{\Psi}}}

\newcommand{\bTh}{{\boldsymbol{\Theta}}}
\newcommand{\obTh}{\bar{\boldsymbol{\Theta}}}
\newcommand{\bThh}{{\hat{\bTh}}}

\newcommand{\La}{\Lambda}
\newcommand{\bLa}{{\boldsymbol{\La}}}

\newcommand{\bSi}{{\boldsymbol{\Sigma}}}

\newcommand{\Ups}{\Upsilon}
\newcommand{\bUp}{{\boldsymbol{\Ups}}}

\newcommand{\bOm}{{\boldsymbol{\Omega}}}
\newcommand{\cOm}{{\mathcal{ K}}_\Omega}
\newcommand{\Wh}{\hat{\Omega}}

%% file: gww_rvs.tex
\DeclareMathAlphabet{\mathbsf}{OT1}{cmss}{bx}{n}% bold sans serif
\DeclareMathAlphabet{\mathssf}{OT1}{cmss}{m}{sl}% slanted sans serif

\DeclareSymbolFont{bsfletters}{OT1}{cmss}{bx}{n}  
\DeclareSymbolFont{ssfletters}{OT1}{cmss}{m}{n}
\DeclareMathSymbol{\bsfGamma}{0}{bsfletters}{'000}
\DeclareMathSymbol{\ssfGamma}{0}{ssfletters}{'000}
\DeclareMathSymbol{\bsfDelta}{0}{bsfletters}{'001}
\DeclareMathSymbol{\ssfDelta}{0}{ssfletters}{'001}
\DeclareMathSymbol{\bsfTheta}{0}{bsfletters}{'002}
\DeclareMathSymbol{\ssfTheta}{0}{ssfletters}{'002}
\DeclareMathSymbol{\bsfLambda}{0}{bsfletters}{'003}
\DeclareMathSymbol{\ssfLambda}{0}{ssfletters}{'003}
\DeclareMathSymbol{\bsfXi}{0}{bsfletters}{'004}
\DeclareMathSymbol{\ssfXi}{0}{ssfletters}{'004}
\DeclareMathSymbol{\bsfPi}{0}{bsfletters}{'005}
\DeclareMathSymbol{\ssfPi}{0}{ssfletters}{'005}
\DeclareMathSymbol{\bsfSigma}{0}{bsfletters}{'006}
\DeclareMathSymbol{\ssfSigma}{0}{ssfletters}{'006}
\DeclareMathSymbol{\bsfUpsilon}{0}{bsfletters}{'007}
\DeclareMathSymbol{\ssfUpsilon}{0}{ssfletters}{'007}
\DeclareMathSymbol{\bsfPhi}{0}{bsfletters}{'010}
\DeclareMathSymbol{\ssfPhi}{0}{ssfletters}{'010}
\DeclareMathSymbol{\bsfPsi}{0}{bsfletters}{'011}
\DeclareMathSymbol{\ssfPsi}{0}{ssfletters}{'011}
\DeclareMathSymbol{\bsfOmega}{0}{bsfletters}{'012}
\DeclareMathSymbol{\ssfOmega}{0}{ssfletters}{'012}

\renewcommand{\defeq}{\triangleq}

\newcommand{\rvbH}{{\mathbsf{H}}}

\newcommand{\rvu}{{\mathssf{u}}}	% u

\newcommand{\rvw}{{\mathssf{w}}}	% w

\newcommand{\rvx}{{\mathssf{x}}}	% x, random variable

\newcommand{\rvbx}{{\mathbsf{x}}}

\newcommand{\rvy}{{\mathssf{y}}}	% y

\newcommand{\rvby}{{\mathbsf{y}}}
\newcommand{\rvbyr}{{\mathbsf{y}}_\mathrm{r}}
\newcommand{\rvbye}{{\mathbsf{y}}_\mathrm{e}}

\newcommand{\rvbz}{{\mathbsf{z}}}
\newcommand{\rvbzr}{{\mathbsf{z}}_\mathrm{r}}
\newcommand{\rvbze}{{\mathbsf{z}}_\mathrm{e}}

\newcommand{\hrvbzr}{{\mathbsf{\hat{z}}}_\mathrm{r}}

\newcommand{\hrvbyr}{{\mathbsf{\hat{y}}}_\mathrm{r}}